\documentclass{theoretics}
\title{Completing the Picture for the Skolem Problem on Order-4 Linear Recurrence Sequences}

\ThCSauthor[oxford,mpi]{Piotr Bacik}{piotr.bacik@stcatz.ox.ac.uk}[0009-0006-0248-3204]
\ThCSaffil[oxford]{University of Oxford, United Kingdom}
\ThCSaffil[mpi]{Max Planck Institute for Software Systems, Saarland Informatics Campus, Saarbrücken, Germany}
\ThCSthanks{The author was supported by EPSRC grant EP/X033813/1.}
\ThCSshortnames{P. Bacik}  
\ThCSshorttitle{Skolem Problem on Order-4 Linear Recurrence Sequences}
\ThCSyear{2025}
\ThCSarticlenum{28}
\ThCSreceived{Sep 5, 2024}
\ThCSrevised{Jul 11, 2025}
\ThCSaccepted{Oct 21, 2025}
\ThCSpublished{Dec 2, 2025}
\ThCSdoicreatedtrue
\ThCSkeywords{Skolem Problem, linear recurrence sequences, linear forms in logarithms, algebraic, number theory, exponential polynomial}

\usepackage{csquotes}
\usepackage{amsfonts}
\usepackage{faktor}
\usepackage{bm}

\usepackage[capitalise,noabbrev,nameinlink]{cleveref}

\newcommand{\Q}{\mathbb{Q}}

\newcommand{\p}{\mathfrak{p}}
\newcommand{\Z}{\mathbb{Z}}
\newcommand{\N}{\mathbb{N}}
\renewcommand{\O}{\mathcal{O}}
\newcommand{\LRS}[1]{{\bm{#1}}}

\begin{document}
\maketitle
\begin{abstract}
For almost a century, the decidability of the Skolem Problem, that is, the problem of determining whether a given linear recurrence sequence (LRS) has a zero term, has remained open. A breakthrough in the 1980s established that the Skolem Problem is decidable for algebraic LRS of order at most 3, and real algebraic LRS of order at most 4. However, for general algebraic LRS of order 4 decidability has remained open. Our main contribution in this paper is to prove decidability for this last case, i.e.\ we show that the Skolem Problem is decidable for all algebraic LRS of order at most 4. 
\end{abstract}
\section{Introduction}
Given a ring $R$, a \emph{linear recurrence sequence} over $R$ ($R$-LRS for short) is a sequence $\LRS{u} = \langle u_n \rangle_{n=0}^\infty$ of elements of $R$ satisfying a linear recurrence relation of the form
\begin{align} \label{LRS_def}
	u_{n+d} = a_{d-1} u_{n+d-1} + \dots + a_0 u_n
\end{align}
where $a_0, \dots , a_{d-1} \in R$ and $a_0 \neq 0$. If $d$ is minimal such that a relation of the form \eqref{LRS_def} holds, we call $d$ the \emph{order} of $\LRS{u}$. The celebrated Skolem-Mahler-Lech theorem \cite{Skolem_SML,Mahler_SML,lech_note_1953} states that the zero set $\{n \in \Z_{\geq 0} : u_n = 0\}$ of an $R$-LRS is the union of finitely many arithmetic progressions and a finite set, when $R$ is an integral domain of characteristic zero. 

Unfortunately, all known proofs of the Skolem-Mahler-Lech theorem are ineffective - there is no known method to compute the zeroes of a given algebraic LRS (that is, a $\overline \Q$-LRS, where $\overline \Q$ is the field of algebraic numbers). Equivalently there is currently no known method to decide the problem of whether a given arbitrary $\overline \Q$-LRS has a zero. This problem is known as the Skolem Problem. 

So far, only special cases are known, and the general problem has remained open for around 90 years. The Skolem Problem is closely connected with many topics in theoretical computer science and other areas, including loop termination \cite{ouaknine_linear_2015,almagor_deciding_2021}, formal power series (e.g. \cite[Section 6.4]{berstel_noncommutative_2010}), matrix semigroups \cite{bell_mortality_2021}, stochastic systems \cite{agrawal_approximate_2015,barthe_universal_2020} and control theory \cite{blondel_survey_2000}. 
	
Substantial work has been done on the decidability of the Skolem problem. The breakthrough papers \cite{Vereshchagin,Mignotte_distance} established decidability for algebraic LRS of order at most 3 and real algebraic LRS of order at most 4. To date, this is the state of the art on the Skolem Problem; even in the simplest case of integer sequences, decidability is not yet known for order-5 LRS. 

More recently, decidability has been shown to hold subject to some additional assumptions. Lipton \emph{et al.}\ \cite{Lipton_2022} establish decidability for \emph{reversible} integer LRS (i.e, those sequences that may be continued backwards in $\Z$) of order at most 7. Further, Bilu \emph{et al.}\ \cite{Yuri_2023} show that the Skolem problem is decidable for \emph{simple} $\Q$-LRS\footnote{An LRS is simple if all the roots of its characteristic polynomial are simple.}, assuming the $p$-adic Schanuel conjecture and the Skolem conjecture hold. 
	
Our contribution in this paper is to extend the results of \cite{Vereshchagin,Mignotte_distance} by removing the assumption that the LRS is real; we show decidability for \emph{all} algebraic order 4 LRS. It was stated in \cite[p. 33]{Chonev} that the general algebraic case of order 4 is difficult. However, we in fact show that this case may be solved by a relatively simple recombination of ideas already present in \cite{Vereshchagin,Mignotte_distance}.

 \section{Preliminaries}\label{Prelim}
 Here we briefly summarise some basic notions about algebraic numbers, heights, and LRS. More details may be found in \cite{neukirch_algebraic_1999,Waldschmidt_book}.
\subsection{Algebraic numbers and heights}
 Let $K$ be a finite extension of $\Q$, i.e.\ a number field. Let $\O_K$ denote the ring of algebraic integers in $K$. Define a \emph{fractional ideal} of $K$ to be a non-zero finitely generated $\O_K$-submodule of $K$. Equivalently $I$ is a fractional ideal if and only if there is $c \in K$ such that $cI \subseteq \O_K$ is an ideal of $\O_K$. The fractional ideals of $K$ form a group under multiplication. We have a unique factorisation theorem for fractional ideals \cite[p. 22]{neukirch_algebraic_1999}.
 \begin{theorem}[Unique Factorisation Theorem] \label{thm:ideal_fac}
Any fractional ideal $I$ of $K$ has a unique decomposition
\begin{align*}
I = \prod_{i=1}^t \p_i^{n_i}
\end{align*}
where $t \in \N$, each $\p_i \subseteq \O_K$ is a prime ideal and $n_i \in \Z$. 
 \end{theorem}

 For a prime ideal $\p \subseteq \O_K$ we define the \emph{valuation} $v_\p : K \to \Z \cup \{\infty\}$ by $v_\p(0) = \infty$ and for non-zero $a \in K$, we define $v_\p(a)$ to be the exponent of $\p$ in the prime decomposition of the fractional ideal $a\O_K$ given by Theorem \ref{thm:ideal_fac}. If $p \in \mathbb Z$ is a prime, we may define the $p$-adic absolute value on $\mathbb Q$ by $|x|_p = p^{-v_p(x)}$. By Ostrowski's theorem, any (non-trivial) absolute value on $\Q$ is equivalent to the ``usual'' absolute value $|\cdot|$ or a $p$-adic absolute value $|\cdot|_p$ for some prime $p \in \Z$. Furthermore, if $|\cdot|_v$ is a (non-trivial) absolute value on $K$, then its restriction to $\Q$ is either equivalent to the "usual" absolute value $|\cdot|$ (in which case $|\cdot|_v$ is \emph{Archimedean}) or some $p$-adic absolute value $|\cdot|_p$ (in which case $|\cdot|_v$ is \emph{non-Archimedean} and we write $v \mid p$). We normalise $|\cdot|_v$ such that, when restricting to $\Q$, $|\cdot|_v$ coincides with the usual absolute value $|\cdot|$ or $|\cdot|_p$ for some $p$; in the Archimedean case we have $|x|_v = x$ for all $x \in \Q_{\geq 0}$, and in the non-Archimedean case $|p|_v = p^{-1}$. Denote the set of non-trivial absolute values on $K$, normalised as above, by $M_K$.

If $v \in M_K$ is Archimedean, it either corresponds to a real embedding $K \overset{\sigma}\hookrightarrow \mathbb R$ or a pair of complex embeddings $K \overset{\sigma,\overline \sigma}\hookrightarrow \mathbb C$. In both cases, under the normalisation above we have $|x|_v = |x|_\sigma = |\sigma(x)|$. Define the \emph{local degree} $d_v$ of $v$ as $d_v = 1$ or $d_v = 2$ if $v$ corresponds to a real embedding or a pair of complex conjugate embeddings respectively.

If $v \in M_K$ is non-Archimedean, and $v \mid p$, then $v$ corresponds to a prime ideal $\p \subseteq \O_K$ dividing $p$. Under the normalisation above, we have $|x|_v = |x|_\p = p^{-v_\p(x)/e_\p}$ where $e_{\p}$ is the ramification index; the exponent of $\p$ in the prime ideal decomposition of $p$ in $\O_K$. Define $d_v = [K_v:\Q_p]$ where $K_v,\Q_p$ are the completions of $K, \Q$ with respect to $|\cdot|_v, |\cdot|_p$ respectively.

If $x \in K$ and $[K:\Q] = D$, we define the \emph{absolute logarithmic height} of $x$ (or just \emph{height} for short) as
\begin{align*}
    h(x) := \frac{1}{D} \sum_{v \in M_K} d_v \log \max\{|x|_v,1\}.
\end{align*}

It is known that the height depends only on $x$, not the choice of number field $K$. The height satisfies the following properties, found in \cite[Sections 3.2, 3.5]{Waldschmidt_book}.
\begin{proposition} \label{prop:h_props}
For any algebraic numbers $\alpha_1,\alpha_2$ with $\alpha_1 \neq 0$, for any $n \in \Z$ and any absolute value $|\cdot|_v$ on $\Q(\alpha_1)$ we have
\begin{enumerate}
    \item $h(\alpha_1 \alpha_2) \leq h(\alpha_1) + h(\alpha_2)$, \label{enum:h_props_mult}
    \item $h(\alpha_1 + \alpha_2) \leq \log 2 + h(\alpha_1) + h(\alpha_2)$, \label{enum:h_props_add}
    \item $h(\alpha_1^n) = |n|h(\alpha_1)$. \label{enum:h_props_pow}
    \item $-[\Q(\alpha_1):\Q]h(\alpha_1) \leq \log |\alpha_1|_v \leq [\Q(\alpha_1):\Q]h(\alpha_1)$ \label{enum:h_props_log}
\end{enumerate}
\end{proposition}
\subsection{Linear recurrence sequences} \label{Linear Recurrence Sequences}
 Consider a  $\overline{\mathbb Q}$-LRS $\LRS{u}$ of order $d$ satisfying \eqref{LRS_def}; define its \emph{characteristic polynomial} as $g(X) = X^d - a_{d-1}X^{d-1} - \dots - a_0$, and call the distinct roots $\{\lambda_1, \dots , \lambda_s\}$ of $g$ the \emph{characteristic roots} of $\LRS{u}$. If $\lambda_i/\lambda_j$ is a root of unity for some $i \neq j$, call $\LRS{u}$ degenerate, otherwise say $\LRS{u}$ is non-degenerate. By \cite[Theorem 1.2]{recurrence}, any $\overline \Q$-LRS may be effectively decomposed into a finite number of subsequences which are all either identically zero or non-degenerate. Therefore, the Skolem Problem reduces to the non-degenerate case.
 
 For $K$ a number field containing $\LRS{u}$ and its characteristic roots, given any absolute value $|\cdot|_v$ on $K$, if $|\lambda_i|_v \geq |\lambda_j|_v$ for all $j$ then we say $\lambda_i$ is \emph{dominant with respect to} $|\cdot|_v$. It is well known that $\LRS{u}$ admits an \emph{exponential polynomial} representation
 \begin{align*}
     u_n = \sum_{i=1}^s P_i(n) \lambda_i^n
 \end{align*}
 where each $P_i \in K[X]$ is a polynomial with degree one less than the multiplicity of $\lambda_i$ as a root of $g$.
\section{Linear forms in logarithms and the MSTV class} 
In this section we give a brief overview of the results of \cite{Mignotte_distance,Vereshchagin} and the methods used therein. 
\subsection{Linear forms in logarithms}
The fundamental result used in these papers is Baker's theorem on linear forms in logarithms \cite{Baker_Sharpening_II_1973}, and its analogue for non-Archimedean absolute values. This yields a lower bound on expressions of the form $|\alpha_1^{b_1} \dots \alpha_s^{b_s}-1|_v$ for algebraic numbers $\alpha_1, \dots , \alpha_s$ and integers $b_1, \dots , b_s$ such that $\alpha_1^{b_1} \dots  \alpha_s^{b_s} - 1 \neq 0$. The original inequality, proven for the usual absolute value on $\mathbb C$, has since been improved upon, with various forms existing in the literature. For $|\cdot|_v$ the usual absolute value on $\mathbb C$, we present the following formulation by Matveev \cite[Corollary 2.3]{matveev_explicit_2000}.
\begin{theorem}[Matveev]\label{Matveev} Let $\alpha_1, \dots , \alpha_s$ be non-zero complex algebraic numbers contained in a number field of degree $D$ and $\log \alpha_1, \dots , \log \alpha_s$ some determination of their logarithms. Let $b_1, \dots , b_s \in \Z$ be such that 
\begin{align*}
    \Lambda = b_1 \log \alpha_1, + \dots + b_s \log \alpha_s \neq 0 \, .
\end{align*}
Further, let $A_1, \dots , A_s, B$ be real numbers such that 
\begin{align*}
    &A_j \geq \max\{Dh(\alpha_j), |\log \alpha_j|, 0.16\} \, \, \, (j = 1, \dots , s) \\
    &B = \max\{|b_1|, \dots , |b_s|\} \, .
\end{align*}
Then 
\begin{align*}
    |\Lambda| \geq \exp\left(-2^{6s+20}D^2 A_1 \dots A_s (1+ \log D)(1+ \log B)\right) \, .
\end{align*}
\end{theorem}

For $|\cdot|_v$ non-Archimedean, we have the following result of Yu \cite[Theorem 1]{yu_p-adic_1999}.\footnote{The original papers \cite{Vereshchagin} and \cite{Mignotte_distance} relied on results of a paper \cite{Van_der_Poorten_1977} by Van der Poorten for the non-Archimedean case. However, this article is known to have substantial flaws, subsequently discovered and corrected by Yu.}
\begin{theorem}[Yu]\label{Yu}
Let $\alpha_1, \dots , \alpha_s$ be non-zero algebraic numbers contained in a number field $K$ of degree $D$. Let $\p \subseteq \O_K$ be a prime ideal lying above integer prime $p$, such that $v_\p(\alpha_i) = 0$ for $i = 1, \dots, s$. Let $b_1,\dots,b_s \in \Z$ be such that
\begin{align*}
    \Xi = \alpha_1^{b_1} \dots \alpha_s^{b_s}-1 \neq 0 \, .
\end{align*}
Further, let $A_1, \dots , A_s, B$ be real numbers such that
\begin{align*}
    &A_j \geq \max\{h(\alpha_j), \log p\} \, \, \, (j=1 , \dots , s) \\
    &B \geq \max \{|b_1| , \dots , |b_s|, 3\}.
\end{align*}
Then
\begin{align*}
    v_\p(\Xi) \leq CA_1 \dots A_s \log B
\end{align*}
where $C>0$ is some constant depending effectively on $s, D$ and  $\p$.
\end{theorem}
\subsection{The MSTV class}
The crux of the approach in \cite{Mignotte_distance,Vereshchagin} is to use Baker's theory of linear forms in logarithms to prove lower bounds on exponential sums. For an excellent expository note which elaborates on all the details of the results of \cite{Mignotte_distance,Vereshchagin}, see \cite{bilu_skolem_2025}. Here, we intend only to give a brief overview. In particular, the following results are key.
\begin{theorem} \label{thm:2_sum}
Let $b_1,b_2$ and $\alpha_1,\alpha_2$ be algebraic numbers in a number field $K$ of degree $D$ such that $\alpha_1/\alpha_2$ is not a root of unity. Let $h = \max\{h(b_1), h(b_2)\}$, and let $|\cdot|_v$ be any non-trivial absolute value (let the underlying prime be $p$ if this absolute value is non-Archimedean). Then there are constants $C_1,C_2 > 0$ depending effectively on $p,D, \alpha_1, \alpha_2$ such that for all $n > C_1 (h+1)$ we have
\begin{align} \label{eqn:2_sum}
    |b_1 \alpha_1^n + b_2 \alpha_2^n|_v \geq |\alpha_1|_v^n e^{-C_2 (h+1) \log n}
\end{align}
\end{theorem}
\begin{theorem} \label{thm:3_sum}
Let $b_1,b_2,b_3$ and $\alpha_1,\alpha_2,\alpha_3$ be algebraic numbers in a number field $K$ of degree $D$ such that no quotient $\alpha_i/\alpha_j$ is a root of unity for $1 \leq i < j \leq 3$. Let $h = \max\{h(b_1), h(b_2),h(b_3)\}$, and let $|\cdot|_v$ be an Archimedean absolute value. Suppose also that $|\alpha_1|_v = |\alpha_2|_v = |\alpha_3|_v$. Then there are constants $C_1,C_2 > 0$ depending effectively on $p,D, \alpha_1, \alpha_2$ such that for all $n > C_1 (h+1)$ we have
\begin{align}
    |b_1 \alpha_1^n + b_2 \alpha_2^n + b_3 \alpha_3^n|_v \geq |\alpha_1|_v^n e^{-C_2 (h+1) \log n}
\end{align}
\end{theorem}
In the case of Archimedean absolute value, Theorem \ref{thm:2_sum} is \cite[Corollary 2]{Mignotte_distance}, and Theorem \ref{thm:3_sum} follows from combining \cite[Theorems 4, 5]{Mignotte_distance}. The non-Archimedean case of Theorem \ref{thm:2_sum} follows from the $p$-adic analogue of Baker's theory; for completeness we detail a proof as the exact statement is not found in \cite{Mignotte_distance}.
\begin{proof}[Proof of Theorem \ref{thm:2_sum} in the non-Archimedean case]
Let $|\cdot|_v = |\cdot|_\p$ for prime ideal $\p \subseteq \O_K$. First, note that $b_1 \alpha_1^n + b_2 \alpha_2^n = 0$ implies that
\begin{align*}
\left( \frac{\alpha_2}{\alpha_1} \right)^n = -\frac{b_1}{b_2} \, .
\end{align*}
Let $\alpha = \alpha_2/\alpha_1$ and $b = -b_1/b_2$. Then by Proposition \ref{prop:h_props} item \ref{enum:h_props_pow},
\begin{align} \label{eqn:h_quot}
n = \frac{h( b )}{h ( \alpha )} \, .
\end{align}
Now, a result of Voutier \cite[Corollary 2]{PaulVoutier1996} shows that for any algebraic number $\beta$ of degree $D'$ that is not a root of unity, we have
\begin{align} \label{eqn:Voutier}
h(\beta) \geq \frac{2}{D'(\log(3D'))^3} \, .
\end{align}
Applying \eqref{eqn:Voutier} with $\beta = \alpha$, and using items \ref{enum:h_props_mult}, \ref{enum:h_props_pow} from Proposition \ref{prop:h_props} on $h(b) = h(-b_1/b_2)$, equation \eqref{eqn:h_quot} implies that
\begin{align*}
    n \leq D(\log(3D))^3h \, .
\end{align*}
Therefore, for $n > D(\log(3D))^3h$ we have
\begin{align*}
\Xi =  b^{-1}\alpha^n - 1 \neq 0 \, .
\end{align*}
Note that if $\left| b^{-1}\alpha^n \right|_\p \neq 1$, then $|\Xi|_\p = 1$ and \eqref{eqn:2_sum} follows easily. So we only need to consider when $\left|b^{-1}\alpha^n \right|_\p = 1$. Suppose first that $\left|\alpha \right|_\p \neq 1$. In this case, since $\log |\alpha|_\p \neq 0$, we have 
\begin{align*}
    n = \frac{\log \left| b \right|_\p}{\log |\alpha|_\p} \, .
\end{align*}
Since $|\alpha|_\p \neq 1$, either $|\alpha|_\p \geq p^{1/e_{\p}}$ or $|\alpha|_\p \leq p^{-1/e_{\p}}$, meaning $|\log |\alpha|_\p| \geq \frac{1}{e_{\p}} \log p$. Thus, with items \ref{enum:h_props_mult}, \ref{enum:h_props_pow}, and \ref{enum:h_props_log} of Proposition \ref{prop:h_props} applied to $\log|b|_\p$ and $h(b)$, we get
\begin{align*}
    n \leq \frac{e_{\p}}{\log p}D h(b) \leq \frac{2e_{\p}}{\log p} D h 
\end{align*}
and so $|\Xi|_\p = 1$ for $n > \frac{2e_{\p}}{\log p} D h $.

Otherwise, $\left|\alpha \right|_\p = 1$ and therefore $|b|_\p = 1$, so we may apply Theorem \ref{Yu} to $\Xi$ to get that for all $n > \max\{D(\log(3D))^3h,3\}$
\begin{align}
    v_\p(\Xi) &\leq C \max\left\{h( b),\log p \right\} \max \left\{ h( \alpha) , \log p \right\} \log n \notag \\
    &\leq C' (h+1) \log n \label{eqn:vp_Xi}
\end{align}
for constants $C,C'$ depending only on $\p,D,\alpha_1,\alpha_2$. Thus, let $C_1 = \max \left\{D(\log(3D))^3,\frac{2e_{\p}}{\log p} D,3 \right\}$, then for all $n > C_1(h+1)$ we have
\begin{align*}
|b_1 \alpha_1^n + b_2 \alpha_2^n|_\p &\geq |b_1 \alpha_1^n|_\p |\Xi|_\p \\
&\geq |\alpha_1|_\p^n e^{-hD} p^{-v_\p(\Xi)/e_{\p}} \\
&\geq |\alpha_1|^n_\p e^{-C_2(h+1)\log n}
\end{align*}
for $C_2>0$ depending only on $\p,D,\alpha_1,\alpha_2$, where in the second inequality we used Proposition \ref{prop:h_props} item \ref{enum:h_props_log}, and in the third we used \eqref{eqn:vp_Xi}. Note we can take $C_2$ to depend on $p$ rather than $\p$ by taking the minimum of the constants over the finitely many prime ideals of $\O_K$ lying above $p$.
\end{proof}

We exhibit how Theorems \ref{thm:2_sum} and \ref{thm:3_sum} lead to decidability for a large class of LRS.
\begin{definition}[MSTV class]
The Mignotte-Shorey-Tijdeman-Vereshchagin (MSTV) class consists of all $\overline \Q$-LRS that have at most 3 dominant roots with respect to some Archimedean absolute value, or at most 2 dominant roots with respect to some non-Archimedean absolute value.
\end{definition}
Note that this terminology was introduced in \cite{Lipton_2022} for $\mathbb Z$-LRS, our definition subsumes the one given there. 
\begin{theorem}\label{MSTV_Decidable}
The Skolem problem is decidable for all non-degenerate LRS in the MSTV class.
\end{theorem}

\begin{proof}
Let $\LRS{u}$ be a non-degenerate LRS in the MSTV class with distinct characteristic roots $\lambda_1, \dots , \lambda_s$. For some absolute value $|\cdot|_v$ and for some $1 \leq r \leq 3$ we have 
\begin{align*}
|\lambda_1|_v = \dots = |\lambda_r|_v > |\lambda_{r+1}|_v \geq \dots \geq |\lambda_s|_v \, .
\end{align*}
As noted in Section \ref{Linear Recurrence Sequences}, $\LRS{u}$ admits the exponential polynomial representation
\begin{align}\label{Exp_sum}
    u_n = \sum_{i=1}^s P_i(n) \lambda_i^n 
\end{align}
for polynomials $P_1, \dots , P_s$ with algebraic coefficients. When $|\cdot|_v$ is any non-trivial absolute value and $r = 2$, or when $|\cdot|_v$ is Archimedean and $r=3$, we apply Theorem \ref{thm:2_sum} or \ref{thm:3_sum} to $\sum_{i=1}^r P_i(n) \lambda_i^n$ to get computable constants $C_1,C_2 \geq 0$ such that for all $n > C_1(h+1)$
\begin{align*}
\left| \sum_{i=1}^r P_i(n) \lambda_i^n \right|_v \geq |\lambda_1|_v^n e^{-C_2 \left(\underset{i} \max \{h(P_i(n))\} + 1 \right) \log n} \, .
\end{align*}
Note also that for $r=1$ such an inequality is trivial. Moreover, by repeated application of items \ref{enum:h_props_mult},\ref{enum:h_props_add} of Proposition \ref{prop:h_props}, we have $\underset{i}{\max} \{h(P_i(n))\} \leq C_3 \log n$ for $n\geq 2$ and for some effective constant $C_3 >0$ depending on the coefficients of each $P_i$. Thus, we have for all $n > C_1(h+1)$
\begin{align} \label{eqn:final_ineq}
|u_n| \geq \left| \sum_{i=1}^r P_i(n) \lambda_i^n \right|_v - \left| \sum_{i=r+1}^s P_i(n) \lambda_i^n \right|_v \geq |\lambda_1|_v^n e^{-C_4 (\log n)^2} - C_5 n^{C_6}|\lambda_{r+1}|_v^n
\end{align}
for some effective $C_4,C_5,C_6 > 0$. Since $|\lambda_1|_v > |\lambda_{r+1}|_v$, the first term of the right-hand side of \eqref{eqn:final_ineq} grows faster than the second and one may therefore compute $C_7 >0$ such that $|u_n| > 0$ for all $n > C_7$. Decidability of the Skolem Problem immediately follows as then one only needs to check whether $u_n = 0$ for each $0 \leq n \leq C_7$.
\end{proof}

In \cite{Lipton_2022} it is written that the Skolem Problem is known to be decidable for \emph{all} $\Z$-LRS in the MSTV class. This is not true. This is because Theorems \ref{thm:2_sum} and \ref{thm:3_sum} only apply to \emph{non-degenerate} LRS in the MSTV class and, while every LRS can be written as the interleaving of non-degenerate subsequences, crucially the MSTV class is not closed under taking subsequences. Indeed, let $\LRS{v}$ be an arbitrary non-degenerate LRS, and let $\lambda$ be an algebraic number whose modulus is larger than any characteristic root of $\LRS{v}$. Defining 
\begin{align*}
    u_n = \lambda^n - (-\lambda)^n + v_n \, ,
\end{align*}
we see that $\LRS{u}$ has 2 dominant roots in modulus, yet $u_{2n} = v_{2n}$. Since $\LRS{v}$ was arbitrary, we see $\langle u_{2n} \rangle_{n=0}^\infty$ need not lie in the MSTV class.

Fortunately, it is easy to see that the main results of \cite{Lipton_2022} (decidability of low order reversible sequences, and order 5 sequences assuming the Skolem Conjecture) is not affected by this oversight.

\section{The Skolem Problem at order 4}
Our aim is to prove the following:
\begin{theorem} \label{MainThm}
	The Skolem problem is decidable for all algebraic LRS with at most $4$ distinct characteristic roots. In particular, decidability holds for all algebraic LRS of order $d \leq 4$.
\end{theorem}
For this we require the following basic result of Kronecker \cite{Kronecker}:
\begin{theorem}[Kronecker] \label{Kronecker}
	Let $\alpha$ be a non-zero algebraic integer. If all Galois conjugates of $\alpha$ lie in the unit disc $\{z \in \mathbb C : |z| \leq 1\}$ then $\alpha$ is a root of unity.
\end{theorem}
Theorem \ref{MainThm} follows immediately from the following lemma.
\begin{lemma} \label{MainLemma}
Given algebraic numbers $\lambda_1, \lambda_2, \lambda_3, \lambda_4$, let $K$ be the Galois closure of $\Q(\lambda_1, \lambda_2, \lambda_3, \lambda_4)$. If $\lambda_1, \lambda_2, \lambda_3, \lambda_4$ are all dominant with respect to every Archimedean absolute value on $K$, and for every non-Archimedean absolute value $|\cdot|_v$ at least 3 of the $\lambda_i$ are dominant with respect to $|\cdot|_v$, then for all $i \neq j$, we have $\lambda_i/\lambda_j$ is a root of unity. 
\end{lemma}
\begin{proof}
	Since all the $\lambda_i$ are dominant with respect to every Archimedean absolute value on $K$, we have
	\begin{align} \label{sigma_equality}
		|\sigma(\lambda_1)| = |\sigma(\lambda_2)| = |\sigma(\lambda_3)| = |\sigma(\lambda_4)|
	\end{align}
	for all $\sigma \in \text{Gal}(K/\Q)$ from which it follows (taking $\sigma$ to be the identity) that
	\begin{align*}
		\lambda_1 \overline{\lambda_1} = \lambda_2 \overline{\lambda_2} = \lambda_3 \overline{\lambda_3} = \lambda_4 \overline{\lambda_4}
	\end{align*}
	and so for every prime ideal $\p \subseteq \O_K$ we have
	\begin{align} \label{padic_equality}
		|\lambda_1|_\p |\overline{\lambda_1}|_\p = |\lambda_2|_\p |\overline{\lambda_2}|_\p = |\lambda_3|_\p |\overline{\lambda_3}|_\p = |\lambda_4|_\p |\overline{\lambda_4}|_\p.
	\end{align}
	For a fixed prime ideal $\p \subseteq \O_K$, by assumption we have (up to relabelling)
	\begin{align} \label{padic_inequality} 
		|\lambda_1|_\p = |\lambda_2|_\p = |\lambda_3|_\p \geq |\lambda_4|_\p.
	\end{align}
	Suppose that exactly 3 of the $\lambda_i$ were dominant with respect to $|\cdot|_\p$, so the inequality in \eqref{padic_inequality} was strict. Then together with \eqref{padic_equality} this forces
	\begin{align*}
		|\overline{\lambda_4}|_\p > |\overline{\lambda_1}|_\p = |\overline{\lambda_2}|_\p = |\overline{\lambda_3}|_\p.
	\end{align*}
	But the absolute value $|\cdot|_v$ defined by $|x|_v = |\overline{x}|_\p$ is non-Archimedean, and $\lambda_4$ is dominant with respect to $|\cdot|_v$, contradicting our assumption that at least 3 of the $\lambda_i$ are dominant with respect to every non-Archimedean absolute value. So in fact the inequality in \eqref{padic_inequality} must be an equality, i.e.\ for every prime ideal $\p \subseteq \O_K$ we have
	\begin{align*}
		|\lambda_1|_\p = |\lambda_2|_\p = |\lambda_3|_\p = |\lambda_4|_\p.
	\end{align*}
	Thus $v_\p(\lambda_i) = v_\p(\lambda_j)$ for all prime ideals $\p \subseteq \O_K$, so by Theorem \ref{thm:ideal_fac} we have equality of fractional ideals $\lambda_i \O_K = \lambda_j \O_K$. Therefore $\frac{\lambda_i}{\lambda_j} \O_K = \O_K$ and so $\frac{\lambda_i}{\lambda_j}$ is an algebraic integer (in fact, a unit of $\O_K$) for each $i \neq j$. Also, \eqref{sigma_equality} implies that $\left|\sigma\left(\lambda_i/\lambda_j\right)\right| = 1$ for all $\sigma \in \text{Gal}(K/\Q)$. Therefore, we may apply Kronecker's Theorem \ref{Kronecker} to conclude that $\lambda_i/\lambda_j$ is a root of unity.
\end{proof}
\begin{proof}[Proof of Theorem \ref{MainThm}]
	Let $\LRS{u}$ be an algebraic LRS with at most 4 distinct characteristic roots. As noted in Section \ref{Linear Recurrence Sequences}, one can effectively split $\LRS{u}$ into non-degenerate subsequences, and since each subsequence also has at most 4 distinct characteristic roots, we may assume $\LRS{u}$ is non-degenerate. By Lemma \ref{MainLemma} applied to the characteristic roots of $\LRS{u}$ and by non-degeneracy of $\LRS{u}$ we conclude that $\LRS{u}$ lies in the MSTV class and so has decidable Skolem problem by Theorem \ref{MSTV_Decidable}.
\end{proof} 
\begin{remark}
Explicitly, what we have shown is that every non-degenerate algebraic LRS with at most 4 distinct characteristic roots is in the MSTV class.
\end{remark}
\section*{Acknowledgements}
    The author would like to thank Prof.\ Péter Varjú for bringing this problem to his attention, and for his guidance in the writing of this paper. The author would also like to thank Prof.\ James Worrell for many useful comments. Finally, the author is grateful to the reviewers whose detailed comments have substantially improved the quality of the paper. The author was supported by EPSRC grant EP/X033813/1.

    \printbibliography
\end{document}